\documentclass[letterpaper, 10 pt, conference]{ieeeconf}

\IEEEoverridecommandlockouts

\usepackage[normalem]{ulem}
\usepackage{cancel}
\usepackage[dvipsnames]{xcolor}

\usepackage{float}
\usepackage{amsmath}
\usepackage{amsfonts}
\usepackage{amssymb}
\usepackage{breqn}
\usepackage{cite}
\usepackage{hyperref}
\usepackage{etoolbox}
\hypersetup{
  colorlinks=true,
  citecolor=blue,
  linkcolor=blue,
  urlcolor=blue
}

\usepackage{adjustbox}
\usepackage{mathrsfs}
\newtheorem{theorem}{Theorem}
\newtheorem{lemma}{Lemma}
\newtheorem{assumption}{Assumption}

\newtheorem{remark}{Remark}
\newtheorem{definition}{Definition}

\title{\LARGE \bf
Resilient Model-Free Asymmetric Bipartite Consensus for Nonlinear Multi-Agent Systems against DoS Attacks
}

\author{Yi Zhang$^{1}$, Yichao Wang$^{2}$, Junbo Zhao$^{3}$, and Shan Zuo$^{4}$
\thanks{$^{1}$Yi Zhang, $^{2}$Yichao Wang, $^{3}$Junbo Zhao, and $^{4}$Shan Zuo are with the Department of Electrical and Computer Engineering, University of Connecticut, 371 Fairfield Way, U-4157, Storrs, Connecticut 06269-4157, USA. (E-mails: yi.2.zhang@uconn.edu; yichao.wang@uconn.edu;junbo@uconn.edu;shan.zuo@uconn.edu).}
}

\begin{document}

\maketitle
\thispagestyle{empty}
\pagestyle{empty}

\begin{abstract}
In this letter, we study an unified resilient asymmetric bipartite consensus (URABC) problem for nonlinear multi-agent systems with both cooperative and antagonistic interactions under denial-of-service (DoS) attacks. We first prove that the URABC problem is solved by stabilizing the neighborhood asymmetric bipartite consensus error. Then, we develop a distributed compact form dynamic linearization method to linearize the neighborhood asymmetric bipartite consensus error. By using an extended discrete state observer to enhance the robustness against unknown dynamics and an attack compensation mechanism to eliminate the adverse effects of DoS attacks, we finally propose a distributed resilient model-free adaptive control algorithm to solve the URABC problem. A numerical example validates the proposed results. 

\end{abstract}

\section{INTRODUCTION}
\label{sec:introduction}

The cooperative control of the multi-agent system (MAS) has become an in-demand topic in control theory. For distributed coordination, creating information flow protocols is essential for collective consensus. This issue is commonly termed as consensus or synchronization problem \cite{olfati2004consensus,ren2007information,cao2012overview}. Most results assume agents cooperate by default \cite{bu2017data}. However, trust and distrust in social networks and economic duopolies show that cooperation and antagonism often coexist in practical. Besides, constructing accurate models of these real-world systems is labor-intensive and raises issues of unmodeled dynamics and robustness. This makes control design too complex and unsuited for solving engineering issues, especially with network-induced problems \cite{zhang2019networked}. Despite these challenges, model-free adaptive control can address nonlinear systems with completely unknown dynamics \cite{hou2010novel}. 

Denial-of-service (DoS) threats highlight vulnerabilities in digital infrastructures, urging strong cybersecurity solutions. DoS attacks aim to block the transmission of data packets, which can degrade system performance and cause data packet dropouts. Consequently, it is important to explore resilient cooperative control against DoS attacks. To the best of our knowledge, limited research exists on model-free adaptive control for asymmetric bipartite consensus under DoS attacks. Moreover, research specifically targeting the multi-input multi-output (MIMO) MAS is scarce. Motivated by these challenges, in this letter, we study the unified resilient asymmetric bipartite consensus (URABC) problem for nonlinear MIMO MAS under DoS attacks.

The main contributions of this letter are: First, we prove that the URABC problem is solved by stabilizing the
neighborhood asymmetric bipartite consensus error (NABCE). Second, a distributed compact form dynamic linearization (DCFDL) method is designed to linearize the NABCE. Unlike previous methods that require global information \cite{ren2020robust,wang2020consensus}, the proposed DCFDL achieves a balance between cooperative and antagonistic objectives without using any global information. By using an extended discrete state observer (EDSO)
to enhance the robustness against unknown dynamics and an attack compensation mechanism to eliminate the adverse effects of DoS attacks, we finally propose a distributed resilient model-free adaptive control (DRMFAC) algorithm to solve the URABC problem. Compared with existing results \cite{ren2020robust,wang2020consensus}, our DRMFAC algorithm only utilizes input/output data of the nonlinear MIMO MAS without using any mathematical model of the system dynamics. Moreover, in contrast to the approaches in \cite{guo2018asymmetric} and \cite{liang2021finite} that necessitate a strongly connected communication digraph, our approach relaxes such constraint by considering a weakly connected diagraph.


\section{Preliminaries and Problem Formulation}
\label{sec:pre}


In this paper, $\left\|X^{n}\right\|$ and $\left\|X^{m \times n}\right\|$ represent the Euclidean norm and 2-norm, respectively. $\mathbf{1}_N\in\mathbb{R}^N$ is a vector with all entries are one. $\mathbf{0}$ is a vector with all entries are zero. $Q\succ 0$ denotes the matrix $Q$ is positive-definite. $\rho(A)=\mathrm{max}\{|z_1,\cdots,|z_n|\}$ denotes the spectral radius of matrix $A \in \mathbb{R}^{n \times n}$, with eigenvalue $z_l,l=1,\cdots,n$. $\|\cdot\|_{d}$ denotes an induced matrix norm satisfies $\|A x\| \leqslant \|A\|_{d}\|x\|$. $\mathbf{P}\{\mathcal{E}\}$ gives the probability of event $\mathcal{E}$. We consider a MAS consisting of one leader and $N$ followers, where the interactions among the followers are represented by a signed digraph $\mathscr{G}=(\mathscr{V}, \mathscr{E}, \mathcal{A})$, where $\mathscr{V}=\{1,2, \ldots, N\}$ is the set of vertices, $\mathscr{E} \subset \mathscr{V} \times \mathscr{V}$ denotes the set of edges, and $\mathcal{A}=[a_{ij}]\in\mathbb{R}^{N\times N}$ is the associated adjacency matrix where $a_{ij} \neq 0$ if $(j, i) \in \mathscr{E}$. The neighborhood of the agent $i$ is $\mathcal{N}_{i}=\{j\in\mathscr{V}:(j,i)\in\mathscr{E}\}$ and the self-edge $(i,i)$ satisfies $(i,i)\notin\mathscr{E}$. The in-degree matrix is defined as $\mathcal{D}=\operatorname{diag}\left(d_{i}\right)$ with $d_{i}=\sum\nolimits_{j \in \mathcal{N}_{i}}a_{ij}$. The Laplacian matrix ${\mathcal{L}}$ is defined as
$\mathcal{L} = \mathcal{D} - \mathcal{A}$. The pining gain matrix $\mathcal{G} = \mathrm{diag}(g_i)$, where $g_i$ is the pinning gain from the leader to the follower $i$.


\begin{definition}[\cite{altafini2012consensus}] 
The signed graph is structurally balanced if the vertex set $\mathscr{V}$ can be partitioned into $\mathscr{V}_{1}$ and $ \mathscr{V}_{2}$, with $\mathscr{V}_{1} \cup \mathscr{V}_{2}=\mathscr{V}$ and $\mathscr{V}_{1} \cap \mathscr{V}_{2}=\emptyset$, such that $a_{ij} \geq 0, \forall i,j \in \mathscr{V}_{\iota}$ and $a_{ij} \leq 0, \forall i \in \mathscr{V}_{\iota}, j \in \mathscr{V}_{3-\iota}$, where $\iota=\{1,2\}$.
\end{definition}

Consider the following discrete-time nonlinear MIMO MAS for the $N$ followers
\vspace{-5pt} \begin{equation}
\label{eq1}
    y_{i}(k+1)=f_{i}\left(y_{i}(k),u_{i}(k)\right),i \in \mathscr{V}
\end{equation}
where $y_{i}(k) \in \mathbb{R}^p$ and $u_{i}(k) \in \mathbb{R}^q$ are the output and input of follower $i$ at the time instant $k \in\{1,2, \ldots\}$, respectively. $f_{i}(\cdot)$ is an unknown smooth nonlinear function.
\begin{assumption}\label{ass:1}
The signed digraph $\mathcal{G}=(\mathscr{V}, \mathscr{E}, \mathcal{A})$ is structurally balanced and contains a spanning tree with the leader as the root.
\end{assumption}


\begin{definition}[Unified Asymmetric Bipartite Consensus]\label{obj:1}
Given a desired reference value $y_{d}$ issued by the leader, the unified asymmetric bipartite consensus objective is to achieve
\vspace{-5pt} \begin{equation}
\label{eq2}
\left\{\begin{aligned}
&\lim _{k \to \infty} y_{i}(k)=my_{d},  &i \in \mathscr{V}_{1} \\
&\lim _{k \to \infty} y_{i}(k)=-n y_{d}, &i \in \mathscr{V}_{2}
\end{aligned},\quad \mathscr{V}_{1} \cup \mathscr{V}_{2}=\mathscr{V}\right.
\end{equation}
\end{definition}
where $m$ and $n$ are influence coefficients, which are positive. \eqref{eq2} can be further formulated as
\vspace{-5pt} \begin{equation}
\label{eq3}
\mathop {\lim }_{k \to \infty } {e_{{y_i}}}(k) \equiv \mathop {\lim }_{k \to \infty } {\bar y_i}(k) - {y_d} = 0,i \in \mathscr{V}
\end{equation}
where $e_{y_i}(k)$ is referred as the local asymmetric bipartite consensus error. ${\bar y_i}(k) = \frac{{{y_i}(k)}}{{{s_i}}}$, where $s_{i}= \left\{\begin{aligned}m,& \quad i \in \mathscr{V}_{1} \\
-n, & \quad i \in \mathscr{V}_{2} \end{aligned}\right.$. 

The global form of $e_{y_i}(k)$ is ${e_y}(k) = \bar y(k) - {{\mathbf{1}}_N} \otimes y_d$, where $\bar y (k)=\begin{bmatrix}
    \bar y_1^{\mathrm{T}}(k),\cdots,\bar y_N^{\mathrm{T}}(k)
\end{bmatrix}^{\mathrm{T}}$.
\begin{remark} To achieve leader-following consensus control like \eqref{eq2}, Assumption \ref{ass:1} is a necessary
and sufficient condition. For the asymmetric bipartite consensus objective \eqref{eq2}, the training procedure using lower limb rehabilitation robots for patients is a real-world application \cite{wolbrecht2008optimizing}. The fact that robots give unequal auxiliary torques to both legs establishes a unique scaling and assistance or resistance relation concerning the predicted force or trajectory \cite{hussain2013adaptive}. Thus, delving into ways to achieve the bipartite consensus objective related to an asymmetric state holds significance. Notably, the well-explored bipartite consensus control is a special case of the proposed asymmetric bipartite consensus objective \eqref{eq2} when setting $m=1$ and $n=1$.
\end{remark}
\begin{assumption}\label{ass:2}
Partial derivative $\frac{\partial f_{i}(\cdot, \cdot)}{\partial u_{i}(k)}$ is continuous.
\end{assumption} 

\begin{assumption}\label{ass:3}
The system \eqref{eq1} satisfies the generalized Lipschitz condition, i.e., for time instants $k+1, k \geq 0$ and {$u_{i}(k+1) \neq u_{i}(k)$}, there exists a positive constant $\phi_{i}^y$ such that $\|\Delta y_{i}(k+1)\|\leq \phi_{i}^y\|\Delta u_{i}(k)\|$, where $\Delta y_{i}(k+1)=y_{i}(k+1)-y_{i}(k)$ and $\Delta u_{i}(k)=u_{i}(k)-u_{i}(k-1)$.
\end{assumption}

\begin{remark}
\label{rem:2}
Assumption \ref{ass:2} serves as a broad constraint for the nonlinear system. Assumption \ref{ass:3} suggests that the system output increment is constrained by input increment, which is
satisfied in real-world systems. Considering the controller's design to prevent rapid shifts in \cite{Xu2014}, that is $\Delta u_{i}(k)$ is bound, and that bounded input alterations lead to bounded output changes in \cite{Hou2013}.
\end{remark}

\begin{lemma}[\cite{Hui2019}]\label{lem:1}For any time instant $k$, if system \eqref{eq1} satisfies Assumptions \ref{ass:2} and \ref{ass:3} and $\left\|\Delta u_{i}(k)\right\| \neq 0$, a pseudo-partitioned Jacobian matrix (PPJM) $\Phi_{i}^y(k)$ exists such that system \eqref{eq1} can be written as 
\vspace{-5pt} \begin{equation}
\label{eq4}
\Delta y_{i}(k+1)=\Phi_{i}^y(k) \Delta u_{i}(k)    
\end{equation}
where $\Phi_{i}^y(k)\in \mathbb{R}^{p\times q}$ and $\left\|\Phi_{i}^y(k)\right\| \leqslant  \phi_{i}^y$.
\end{lemma} 
\begin{remark}\label{rem:3}
\eqref{eq4} represents a common linear data model frequently employed in the model-free adpative control algorithm design for MAS.
\end{remark}
\begin{lemma}
[Gershgorin Disk Theorem\cite{bell1965gershgorin}]\label{lem:2}
For a complex matrix $C=\left[c_{i j}\right] \in \mathbb{C}^{n\times n}$, there exists a Gershgorin disk $D_{i}=\bigg\{z\bigg|\bigg.| z-c_{i i}|\leqslant\sum_{j=1, j \neq i}^{n}| c_{i j}| \bigg\}$, {$z \in C, i=1, \ldots, n$}. Then all eigenvalues of $C$ lie in the unions of the circle, i.e., $ z\in D_{C}=\bigcup_{i=1}^{N} D_{i}$.
\end{lemma}

\begin{lemma}[\cite{ortega2000iterative}]\label{lem:3}
For $A \in \mathbb{R}^{n \times n}$, there exists an induced consistent matrix norm $\|A\|_{d} \leqslant \rho(A)+c$, where $\rho(A)$ is the spectral radius of $A$ and $c$ is a positive constant.
\end{lemma}

To achieve the unified asymmetric bipartite consensus objective \eqref{eq2}, we introduce the following NABCE
\vspace{-5pt} \begin{equation}
\label{eq5}
\xi_{i}(k)=\left\{\begin{aligned} 
&\sum_{j \in \mathscr{V}_{1}} a_{ij}\left(y_{j}(k)-\operatorname{sgn}\left(a_{ij}\right) y_{i}(k)\right)
\\&+\sum_{o \in \mathscr{V}_{2}} a_{io}\left(y_{o}(k)-m^{-1}n \operatorname{sgn}\left(a_{io}\right) y_{i}(k)\right)
\\&+g_{i}\left(y_{d}-m^{-1}\operatorname{sgn}\left(g_{i}\right) y_{i}(k)\right),i \in \mathscr{V}_{1}
\\&\sum_{j \in \mathscr{V}_{1}} a_{ij}\left(y_{j}(k)-n^{-1} m\operatorname{sgn}\left(a_{ij}\right) y_{i}(k)\right) 
\\&+\sum_{o \in \mathscr{V}_{2}} a_{io}\left(y_{o}(k)-\operatorname{sgn}\left(a_{io}\right) y_{i}(k)\right) \\&+g_{i}\left(y_{d}-n^{-1} \operatorname{sgn}\left(g_{i}\right) y_{i}(k)\right),i \in \mathscr{V}_{2}
\end{aligned}\right.
\end{equation}

Based on \eqref{eq5}, we then present the following lemma to obtain a necessary and sufficient condition to achieve the unified asymmetric bipartite consensus objective \eqref{eq2}.
\begin{lemma}
\label{lem:4}
Under Assumption \ref{ass:1}, consider MAS \eqref{eq1}, the unified asymmetric bipartite consensus objective \eqref{eq2} is achieved if and only if $\lim\nolimits_{k\to\infty}\xi_i(k)=\mathbf{0}$.
\end{lemma}
\begin{proof}
We introduce an associated asymmetric adjacency matrix $\bar{\mathcal{A}}=[\bar{a}_{ij}]$ with $\bar{a}_{ij}=\frac{a_{ij}}{s_j}$, where $s_{j}= \left\{\begin{aligned}m,& \quad i \in \mathscr{V}_{1} \\
-n, & \quad i \in \mathscr{V}_{2} \end{aligned}\right.$. The corresponding asymmetric in-degree matrix is defined as $\bar{\mathcal{D}}=\operatorname{diag}\left(\bar{d}_{i}\right)$ with $\bar{d}_{i}=\sum\nolimits_{j \in \mathcal{N}_{i}}\bar{a}_{ij}$. The asymmetric Laplacian matrix $\bar{\mathcal{L}}$ is defined as
$\bar{\mathcal{L}} = \bar{\mathcal{D}} - \bar{\mathcal{A}}$.

We then reformulate $\xi_i(k)$ in \eqref{eq5} as
\vspace{-5pt} \begin{equation}
\label{eq6}
\xi_{i}(k)= 
\sum\nolimits_{j \in \mathcal{N}_{i}} \bar{a}_{ij}\left(\bar{y}_{j}(k)-\bar{y}_{i}(k)\right)+g_{i}\left(y_{d}- \bar{y}_{i}(k)\right)
\end{equation}
     The global form of \eqref{eq6} is
\vspace{-5pt} \begin{equation}
\label{eq7}
\begin{aligned}
    \xi(k)=&-\left(\big(\bar{\mathcal{L}}+\mathcal{G}\big)\otimes I_p\right)\left(\bar{y}(k) - {{\mathbf{1}}_N} \otimes y_d\right)
    \\=&-\left(\big(\bar{\mathcal{L}}+\mathcal{G}\big)\otimes I_p\right)e_y(k)
\end{aligned}
\end{equation}
where $\xi(k) = \begin{bmatrix} \xi_1^{\mathrm{T}}(k),\cdots,\xi_N^{\mathrm{T}}(k)\end{bmatrix}^{\mathrm{T}}$. As seen, \eqref{eq2} is achieved if and only if $\lim\limits_{k\to\infty}\xi(k)=\mathbf{0}$.
\end{proof} 

Besides, DoS attacks aim to block the transmission of data packets, consequently reducing system performance and causing data packet dropouts. The data packets the controller receives during DoS attacks are represented as 
\vspace{-5pt} \begin{equation}
\label{eq8}
\bar{\xi}_{i}(k)=H_{i}(k) \xi_{i}(k)
\end{equation} where $H_{i}(k)=\mathrm{diag}(h_{i,r}(k)), r=1,\cdots,p$. If DoS attacks succeed in the $r$-th measurement channel, $h_{i,r}(k)=0$; otherwise, $h_{i,r}(k)=1$. Furthermore, the attack probability follows a Bernoulli distribution, and we presuppose that $\mathbf{P}\left\{h_{i,r}(k)=1\right\}=\bar{h}_{i,r}, \mathbf{P}\left\{h_{i,r}(k)=0\right\}=1-\bar{h}_{i,r}$, where $\bar{h}_{i,r}$ is a random probability value.

Now we introduce the following URABC problem.
\begin{definition}
\label{prb:1}
Given Assumptions \ref{ass:1}, \ref{ass:2}, and \ref{ass:3}, consider MAS \eqref{eq1} under DoS attacks modelled in \eqref{eq8}, {\it the URABC problem} is to develop a DRMFAC algorithm such that the asymmetric bipartite consensus error $e_{y_i}(k)$ in \eqref{eq3} is bounded, i.e., there exists a positive constant $b_{i}$ such that $\left\|e_{y_i}(k)\right\| \leqslant b_{i}$.
\end{definition}

Based on \eqref{eq6}, we obtain  
\vspace{-5pt} \begin{equation}
\label{eq9}
\begin{aligned}
\xi_{i}(k+1)=&
\sum_{j \in \mathcal{N}_{i}} \frac{\bar{a}_{ij}}{s_j}y_{j}(k+1)-\sum_{j \in \mathcal{N}_{i}} \frac{\bar{a}_{ij}}{s_i}y_{i}(k+1)
\\&+g_{i}y_d- \frac{g_{i}}{s_i}y_{i}(k+1)
\end{aligned}
\end{equation}
Then, taking \eqref{eq1} into \eqref{eq9} yields
\vspace{-5pt} \begin{equation}
\label{eq10}
\begin{gathered}
  {\xi _i}(k + 1) = \sum\limits_{j \in {\mathcal{N}_i}} {\frac{{{{\bar a}_{ij}}}}{{{s_j}}}} {f_j}({y_j}(k),{u_j}(k)) - \; \hfill \\
  \sum\limits_{j \in {\mathcal{N}_i}} {\frac{{{{\bar a}_{ij}}}}{{{s_i}}}} {f_i}({y_i}(k),{u_i}(k)) + {g_i}{y_d} - \frac{{{g_i}}}{{{s_i}}}{f_i}({y_i}(k),{u_i}(k)) \hfill \\ 
\end{gathered}
\end{equation}

It can be observed that $\xi_{i}(k+1)$ is a linear combination of multiple nonlinear functions of $u_{i}(k), {y}_{i}(k), u_{j}(k)$, and ${y}_{j}(k)$. Hence, we rewrite $\xi_{i}(k+1)$ as
\vspace{-5pt} \begin{equation}
\label{eq11}
\xi_{i}(k+1)=\bar{f}_{i}\bigg(y_{i}(k), u_{i}(k), y_{j}(k), u_{j}(k)\bigg)
\end{equation}where $\bar {f}_{i}(\cdot)$ is an unknown smooth nonlinear function representing the nonlinear relationship between $\xi_i(k+1)$ and $y_{i}(k)$, $u_i(k)$, $y_{j}(k)$ and $u_j(k)$.

\begin{assumption}\label{ass:4}
For any $k$, $\frac{\left\|\Delta u_{j}(k)\right\|}{\left\|\Delta u_{i}(k)\right\|}<\sigma_{ij}$, $i, j \in \mathscr{V}$, where $\sigma_{ij}$ is a positive constant.
\end{assumption}

We present the following DCFDL for the unified asymmetric bipartite consensus of nonlinear MAS.

\begin{theorem}\label{thm:1}
Given Assumptions \ref{ass:1}, \ref{ass:3}, and \ref{ass:4}, and assume condition $\left\|u_{i}(k)\right\|>\epsilon_{i}$ holds, where $\epsilon_{i}$ is a positive constant, then there exists a PPJM $\Phi_{i}(k)$ such
that \eqref{eq11} is rewritten as
\vspace{-5pt} \begin{equation}
\label{eq12}
\Delta \xi_{i}(k+1)=\Phi_{i}(k) \Delta u_{i}(k)
\end{equation}where $\Delta \xi_{i}(k+1)=\xi_{i}(k+1)-\xi_{i}(k)$ and $\Phi_{i}(k) \in \mathbb{R}^{p \times q }$ is bounded, i.e., $\left\|\Phi_{i}(k)\right\| \leqslant \phi_{i}$, where $\phi_{i}$ is a 
 positive constant.\end{theorem}

\begin{proof}
\label{prf:1}
Based on system \eqref{eq11}, $\Delta \xi_{i}(k+1)$ is written as
\vspace{-5pt} \begin{equation}
\label{eq13}
\begin{aligned}
&\Delta \xi_{i}(k+1)=  \bar{f}_{i}\left(y_{i}(k), u_{i}(k), y_{j}(k), u_{j}(k)\right)
\\&
-\bar{f}_{i}\left(y_{i}(k-1), u_{i}(k-1), y_{j}(k-1), u_{j}(k-1)\right) 
\\&+\bar{f}_{i}\left(y_{i}(k), u_{i}(k-1), y_{j}(k), u_{j}(k)\right) 
\\&-\bar{f}_{i}\left(y_{i}(k), u_{i}(k-1), y_{j}(k), u_{j}(k)\right)
\end{aligned}
\end{equation}

Based on the differential mean value theorem and take the partial derivative of $\bar{f}_{i}\left(y_{i}(k), u_{i}(k), y_{j}(k), u_{j}(k)\right)-\bar{f}_{i}\left(y_{i}(k), u_{i}(k-1), y_{j}(k),u_{j}(k)\right)$ with respect to $u_{i}(k)$ yields
\vspace{-5pt} \begin{equation}
\label{eq14}
\Delta \xi_{i}(k+1)= \bar{F}_{i}^{*}(k) \Delta u_{i}(k)+\Psi_{i}(k)
\end{equation}
where $\bar{F}_i^*(k)=\begin{bmatrix}
    \bar{F}_{i,rs}^*(k)
\end{bmatrix}$ represents the Jacobian matrix of $\bar{f}_{i}$ with respect to $u_{i}(k)$ in the interval $\left[u_{i}(k-1), u_{i}(k)\right]$. $\bar{F}_{i,rs}^*(k)=\frac{\partial \bar{f}_{i,r}^*}{\partial  u_{i,s}(k)},r=1,\cdots,p,s=1,\cdots,q$. $\Psi_{i}(k)=\bar{f}_{i}\left(y_{i}(k), u_{i}(k-1), y_{j}(k), u_{j}(k)\right)-\bar{f}_{i}\left(y_{i}(k-1),u_{i}(k-1), y_{j}(k-1), u_{j}(k-1)\right)$.  

Consider the data equation $\Psi_{i}(k)=E_{i}(k) \Delta u_{i}(k)+$ $E_{j}(k) \Delta u_{j}(k)$ with variable matrix $E_{i}(k)$ and $E_{j}(k)$ for each fixed time $k$. Based on Assumption \ref{ass:4}, there exists one solution $E_{i}^{*}(k)$ such that $\Psi_{i}(k)=E_{i}^{*}(k) \Delta u_{i}(k)$ holds. By setting $\Phi_{i}(k)=$ $\left(\partial \bar{f}_{i}^{*} / \partial u_{i}(k)\right)+E_{i}^{*}(k)$, we derive \eqref{eq12}. Based on Assumptions \ref{ass:3} and  \ref{ass:4}, and using \eqref{eq9}, we have
\begin{flalign}
\label{eq15}
\begin{aligned}
&\|\Delta\xi_{i}(k+1)\|\leq \sum_{j \in \mathcal{N}_{i}}  \left|\frac{\bar{a}_{ij}}{s_j}\right|\|\Delta y_j(k+1)\|\quad\quad\quad
\\&+\sum_{j \in \mathcal{N}_{i}} \left|\frac{\bar{a}_{ij}}{s_i}\right|\|\Delta y_i(k+1)\|+\left|\frac{g_i}{s_i}\right|\|\Delta y_i(k+1)\|
\\&\leq \sum_{j \in \mathcal{N}_{i}}  \phi_{j}^y\left|\frac{\bar{a}_{ij}}{s_j}\right|\|\Delta u_j(k)\|+\sum_{j \in \mathcal{N}_{i}} \phi_{i}^y\left|\frac{\bar{a}_{ij}}{s_i}\right|\|\Delta u_i(k)\|
\\&+\phi_{i}^y\left|\frac{g_i}{s_i}\right|\|\Delta u_i(k)\|
\\&\leq \left(
\sum_{j \in \mathcal{N}_{i}}  \sigma_{ij}\phi_{j}^y\left|\frac{\bar{a}_{ij}}{s_j}\right|+\sum_{j \in \mathcal{N}_{i}} \phi_{i}^y\left|\frac{\bar{a}_{ij}}{s_i}\right|+\phi_{i}^y\left|\frac{g_i}{s_i}\right|\right)\|\Delta u_i(k)\|
\\&\leq\phi_{i}\|\Delta u_i(k)\|
\end{aligned}\raisetag{7\baselineskip}
\end{flalign}
where $\phi_{i}=\sum_{j \in \mathcal{N}_{i}}  \sigma_{ij}\phi_{j}^y\left|\frac{\bar{a}_{ij}}{s_j}\right|+\sum_{j \in \mathcal{N}_{i}} \phi_{i}^y\left|\frac{\bar{a}_{ij}}{s_i}\right|+\phi_{i}^y\left|\frac{g_i}{s_i}\right|$ is a positive constant. Hence, as indicated by \eqref{eq12}, the PPJM $\Phi_{i}(k)$ is bounded and satisfies $\left\|\Phi_{i}(k)\right\| \leqslant \phi_{i}$.
\end{proof}

\section{DRMFAC algorithm design}
\label{sec:design}
In general, the actual value of the time-varying PPJM $\Phi_{i}(k)$ is difficult to obtain. Hence, we define the following performance function to estimate the PPJM $\Phi_{i}(k)$
\vspace{-5pt} \begin{equation}
\label{eq16}
\begin{gathered}
J_{\hat{\Phi}_i}(\hat{\Phi}_{i}(k))=\left(\Delta \xi_{i}(k)-\hat{\Phi}_{i}(k) \Delta u_{i}(k-1)\right)^{\mathrm{T}}Q_i^{\hat{\Phi}}\times
\\ \bigg(\Delta \xi_{i}(k)-\hat{\Phi}_{i}(k) \Delta u_{i}(k-1)\bigg)+\mu_i\left\|\Delta\hat{\Phi}_i(k)\right\|^2
\end{gathered}
\end{equation}
where $Q_i^{\hat{\Phi}}\succeq I$ is a weight matrix, and $\mu_i>0$ is a weight factor. $\hat{\Phi}_{i}(k)\in \mathbb{R}^{p \times q}$ is the estimation of PPJM $\Phi_{i}(k)$. Applying the stationarity condition $\frac{\partial J_{\hat{\Phi}_i}(\hat{\Phi}_{i}(k))}{\partial \hat{\Phi}_{i}(k)} = 0$ to \eqref{eq16} yields
\vspace{-5pt} \begin{equation}
\label{eq17}
\begin{gathered}
-\Delta u_{i}^{\mathrm{T}}(k-1)Q_i^{\hat{\Phi}} \bigg(\Delta \xi_{i}(k)-\hat{\Phi}_{i}(k) \Delta u_{i}(k-1)\bigg)
\\+\mu_i\left\|\Delta\hat{\Phi}_i(k)\right\|=0
\end{gathered}
\end{equation}
Using \eqref{eq17} to design the following update formula for $\hat\Phi_i(k)$
\begin{flalign}
\label{eq18}
\begin{aligned}
\hat{\Phi}_{i}(k)=&\hat{\Phi}_{i}(k-1) +\frac{\eta_{i1}  \left(\xi_{i}(k)-\xi_{i}(k-1)\right) \Delta u_{i}^{\mathrm{T}}(k-1)}{\Delta u_{i}^{\mathrm{T}}(k-1)Q_i^{\hat{\Phi}}\Delta u_{i}(k-1)+\mu_{i}}
\\&-\frac{\eta_{i1} \hat{\Phi}_{i}(k-1) \Delta u_{i}(k-1)\Delta u_{i}^{\mathrm{T}}(k-1)}{\Delta u_{i}^{\mathrm{T}}(k-1)Q_i^{\hat{\Phi}}\Delta u_{i}(k-1)+\mu_{i}}
\end{aligned}\raisetag{2\baselineskip}
\end{flalign}
where $\eta_{i1}\in(0, 2]$ is the step size used to update $\hat{\Phi}_{i}(k)$.

Furthermore, we design the following observer to estimate the NABCE   
\vspace{-5pt} \begin{equation}    \label{eq19}\hat{\xi}_{i}(k+1)=\hat{\xi}_{i}(k)+\hat{\Phi}_{i}(k) \Delta u_{i}(k)+K_{i}\left(\xi_{i}(k)-\hat{\xi}_{i}(k)\right)
\end{equation}
where $K_{i}=\mathrm{diag}(k_{i,r}), k_{i,r}\in(0, 2), r=1,\cdots, p$ is the observer gain matrix. To construct the DRMFAC algorithm, a performance function is defined as
\vspace{-5pt} \begin{equation}
\label{eq20}
J_{u_i}\big(u_{i}(k)\big)=\hat{\xi}_{i}^{\mathrm{T}}(k+1)Q_i^u\hat{\xi}_i(k+1)+\Delta u_{i}^{\mathrm{T}}(k)R_i^u\Delta u_{i}(k)\end{equation}
where $Q_i^u=\mathrm{diag}(\varrho_{i,r}),\varrho_{i,r}>0,r=1,\cdots,p$, and $R_i^u\succ 0$. Substitute \eqref{eq19} to \eqref{eq20} and apply the stationarity condition $\frac{\partial J_{u_i}\big(u_{i}(k)\big)}{\partial \Delta u_{i}(k)} = 0$, we obtain
\vspace{-5pt} \begin{equation}
\label{eq21}
\begin{aligned}
    \Delta u_{i}(k)  = &-\bigg(\hat{\Phi}_{i}(k)^{\mathrm{T}}Q_i^u\hat{\Phi}_{i}(k)+ R_i^u\bigg)^{-1}\times
    \\&\hat{\Phi}_{i}(k)^{\mathrm{T}}Q_i^u\bigg(\hat{\xi}_{i}(k)+K_{i}\left(\xi_{i}(k)-\hat{\xi}_{i}(k)\right)\bigg)
\end{aligned}
\end{equation}

Based on \eqref{eq21}, we design the following iterative formula to update control policy $u_i(k)$
\vspace{-5pt} \begin{equation}
\begin{aligned}
\label{eq22}
    u_i(k) = &u_i(k-1) -\frac{\eta_{i2}\hat{\Phi}_{i}(k)^{\mathrm{T}}Q_i^u}{\|\hat{\Phi}_{i}(k)^{\mathrm{T}}Q_i^u\hat{\Phi}_{i}(k)\|+ \|R_i^u\|}\hfill 
    \\&\times \bigg(\hat{\xi}_{i}(k)+K_{i}\left(\xi_{i}(k)-\hat{\xi}_{i}(k)\right)\bigg)
    \end{aligned}
\end{equation}
where $\eta_{i2}\in(-1, 0)$ is the step size used to update $u_i(k)$.

To mitigate the adverse effects of DoS attacks, we now present the following attack compensation mechanism 
\vspace{-5pt} \begin{equation}
\label{eq23}
\xi_{i}^c(k)=H_{i}(k) \xi_{i}(k)+\left(I-H_{i}(k)\right) \xi_{i}(k-1)
\end{equation}

Substitute \eqref{eq23}  into \eqref{eq18}, \eqref{eq19}, and \eqref{eq22}, we obtain the following DRMFAC algorithm against DoS attacks
\begin{flalign}
\label{eq24}
\begin{aligned}
\hat{\Phi}_{i}(k)=&\hat{\Phi}_{i}(k-1) +\frac{\eta_{i1} \left( \xi_{i}^c(k)-\xi_i(k-1)\right) \Delta u_{i}^{\mathrm{T}}(k-1)}{\Delta u_{i}^{\mathrm{T}}(k-1)Q_i^{\hat{\Phi}}\Delta u_{i}(k-1)+\mu_{i}}
\\&-\frac{\eta_{i1} \hat{\Phi}_{i}(k-1) \Delta u_{i}(k-1)\Delta u_{i}^{\mathrm{T}}(k-1)}{\Delta u_{i}^{\mathrm{T}}(k-1)Q_i^{\hat{\Phi}}\Delta u_{i}(k-1)+\mu_{i}}
\end{aligned}\raisetag{2\baselineskip}
\end{flalign}
\vspace{-5pt} \begin{equation}
\label{eq26}\hat{\xi}_{i}(k+1)=\hat{\xi}_{i}(k)+\hat{\Phi}_{i}(k) \Delta u_{i}(k)+K_{i}\left(\xi_{i}^c(k)-\hat{\xi}_{i}(k)\right)
\end{equation}
\begin{flalign}
\label{eq27}
\begin{aligned}
    u_i(k) = &u_i(k-1) -\frac{\eta_{i2}\hat{\Phi}_{i}(k)^{\mathrm{T}}Q_i^u}{\|\hat{\Phi}_{i}(k)^{\mathrm{T}}Q_i^u\hat{\Phi}_{i}(k)\|+ \|R_i^u\|}\hfill 
    \\&\times \bigg(\hat{\xi}_{i}^c(k)+K_{i}\left(\xi_{i}(k)-\hat{\xi}_{i}(k)\right)\bigg)
    \end{aligned}\raisetag{2.2\baselineskip}
\end{flalign}

We present the following boundedness and resilience analysis of the proposed controller design.

\begin{theorem}\label{thm:1}
Given Assumptions \ref{ass:1}, \ref{ass:2}, \ref{ass:3} and \ref{ass:4} and specified parameters $\mu_i>0$, $\eta_{i1}\in(0,2]$, and $\eta_{i2}\in(-1,0)$. $R_i^u\succ 0$ and $Q_i^{\hat{\Phi}}\succeq 
I$. $Q_i^u=\mathrm{diag}(\varrho_{i,r}),\varrho_{i,r}>0$, and  $K_i=\mathrm{diag}(k_{i,r}), k_{i,r}\in(0,2),r=1,\cdots,p$. For MAS \eqref{eq1} with DoS attacks as in \eqref{eq8}, the URABC problem in Definition \ref{prb:1} is addressed using the DRMFAC protocols \eqref{eq24}, \eqref{eq26}, and \eqref{eq27}.
\end{theorem}

\begin{proof}
\label{prf:2}Define $\tilde{\Phi}_{i}(k)=\hat{\Phi}_{i}(k)-\Phi_{i}(k)$  as the estimation error for $\Phi_{i}(k)$ with $\Delta \xi_{i}(k)=\Phi_{i}(k-1) \Delta u_{i}(k-1)$ in \eqref{eq12},
\begin{flalign}
\label{eq28}
\begin{aligned}
 &\tilde{\Phi}_{i}(k) =\hat{\Phi}_{i}(k-1)-\Phi_{i}(k)
\\&+\frac{\eta_{i1} \left(\xi_{i}^c(k)-\xi_{i}(k-1)\right) \Delta u_{i}^{\mathrm{T}}(k-1)}{\Delta u_{i}^{\mathrm{T}}(k-1)Q_i^{\hat{\Phi}}\Delta u_{i}(k-1)+\mu_{i}}
\\&-\frac{\eta_{i1} \hat{\Phi}_{i}(k-1) \Delta u_{i}(k-1) \Delta u_{i}^{\mathrm{T}}(k-1)}{\Delta u_{i}^{\mathrm{T}}(k-1)Q_i^{\hat{\Phi}}\Delta u_{i}(k-1)+\mu_{i}}
\\ &=\tilde{\Phi}_{i}(k-1)+ \Phi_{i}(k-1)-\Phi_{i}(k)\quad\quad\quad
\\&+\frac{ \eta_{i1}}{\Delta u_{i}^{\mathrm{T}}(k-1)Q_i^{\hat{\Phi}}\Delta u_{i}(k-1)+\mu_{i}}
\\&\times \bigg(H_{i}(k)\xi_{i}(k)+\left(I-H_{i}(k)\right) \xi_{i}(k-1)\bigg.
\\&\bigg.-\xi_{i}(k-1)\bigg)\Delta u_{i}^{\mathrm{T}}(k-1)
\\&-\frac{\eta_{i1} \hat{\Phi}_{i}(k-1) \Delta u_{i}(k-1) \Delta u_{i}^{\mathrm{T}}(k-1)}{\Delta u_{i}^{\mathrm{T}}(k-1)Q_i^{\hat{\Phi}}\Delta u_{i}(k-1)+\mu_{i}}
\\ &=\tilde{\Phi}_{i}(k-1)- \Delta\Phi_{i}(k)
\\&+\frac{ \eta_{i1}\bigg(H_{i}(k)\Delta \xi_{i}(k)-\Delta \xi_{i}(k)\bigg)\Delta u_{i}^{\mathrm{T}}(k-1)}{\Delta u_{i}^{\mathrm{T}}(k-1)Q_i^{\hat{\Phi}}\Delta u_{i}(k-1)+\mu_{i}}
\\&+\frac{ \eta_{i1} \Delta \xi_{i}(k)\Delta u_{i}^{\mathrm{T}}(k-1)}{\Delta u_{i}^{\mathrm{T}}(k-1)Q_i^{\hat{\Phi}}\Delta u_{i}(k-1)+\mu_{i}}
\\&-\frac{\eta_{i1} \hat{\Phi}_{i}(k-1) \Delta u_{i}(k-1) \Delta u_{i}^{\mathrm{T}}(k-1)}{\Delta u_{i}^{\mathrm{T}}(k-1)Q_i^{\hat{\Phi}}\Delta u_{i}(k-1)+\mu_{i}}
\\ &=\tilde{\Phi}_{i}(k-1)- \Delta\Phi_{i}(k)
\\&+\frac{ \eta_{i1}\bigg(H_{i}(k) -I\bigg)\Delta\Phi_{i}(k-1)\Delta u_{i}(k-1)\Delta u_{i}^{\mathrm{T}}(k-1)}{\Delta u_{i}^{\mathrm{T}}(k-1)Q_i^{\hat{\Phi}}\Delta u_{i}(k-1)+\mu_{i}}
\\&-\frac{ \eta_{i1} \Delta\tilde{\Phi}_{i}(k-1)\Delta u_{i}(k-1)\Delta u_{i}^{\mathrm{T}}(k-1)}{\Delta u_{i}^{\mathrm{T}}(k-1)Q_i^{\hat{\Phi}}\Delta u_{i}(k-1)+\mu_{i}}
\end{aligned}\raisetag{14\baselineskip}
\end{flalign}

Denote
$\tilde{\Phi}_{i}(k)=\begin{bmatrix}
\tilde{\Phi}_{i}^{1}(k)^{\mathrm{T}},
\cdots,
\tilde{\Phi}_{i}^{p}(k)^{\mathrm{T}}
\end{bmatrix}^{\mathrm{T}}$, $\Phi_{i}(k)=\begin{bmatrix}
\Phi_{i}^{1}(k)^{\mathrm{T}},
\cdots,
\Phi_{i}^{p}(k)^{\mathrm{T}}
\end{bmatrix}^{\mathrm{T}}$, where $\tilde{\Phi}_{i}^{s}(k)=\begin{bmatrix}\tilde{\Phi}_{i, 1}^s, \cdots\tilde{\Phi}_{i, q}^s\end{bmatrix} \in \mathbb{R}^{1 \times q}$,
$\Phi_{i}^{s}(k)=\begin{bmatrix}\Phi_{i, 1}^s, \cdots, \Phi_{i, q}^s \end{bmatrix} \in \mathbb{R}^{1 \times q}, s=1, \ldots, p$, we have
\begin{flalign}
\label{eq29}
{\begin{aligned}
&\tilde{\Phi}_{i}^{s}(k)= \tilde{\Phi}_{i}^{s}(k-1)\left(I-\frac{\eta_{i1} \Delta u_{i}(k-1) \Delta u_{i}^{\mathrm{T}}(k-1)}{\Delta u_{i}^{\mathrm{T}}(k-1)Q_i^{\hat{\Phi}}\Delta u_{i}(k-1)+\mu_{i}}\right)
\\&+\big(h_{i,r}(k)-1\big)\frac{\eta_{i1} \Phi_{i}^{s}(k-1)\Delta u_{i}(k-1)\Delta u_{i}^{\mathrm{T}}(k-1) }{\Delta u_{i}^{\mathrm{T}}(k-1)Q_i^{\hat{\Phi}}\Delta u_{i}(k-1)+\mu_{i}}
\\&-\Delta \Phi_{i}^{s}(k)
\end{aligned}}\raisetag{2.5\baselineskip}  
\end{flalign}

Since $\eta_{i1} \in(0,2]$, $Q_i^{\hat{\Phi}}\succeq 
I$, and $\mu_{i}>0$,
$\frac{\left\|\eta_{i1} \Delta u_{i}(k-1) \Delta u_{i}^{\mathrm{T}}(k-1)\right\|}{\Delta u_{i}^{\mathrm{T}}(k-1)Q_i^{\hat{\Phi}}\Delta u_{i}(k-1)+\mu_{i}}<2$. Based on \eqref{eq29}, there exists a constant $\sigma_{i} \in (0,1)$ such that
\begin{flalign}
\label{eq30}
\begin{aligned}&\left\|\tilde{\Phi}_{i}^{s}(k-1)\left(I-\frac{\eta_{i1} \Delta u_{i}(k-1) \Delta u_{i}^{\mathrm{T}}(k-1)}{\Delta u_{i}^{\mathrm{T}}(k-1)Q_i^{\hat{\Phi}}\Delta u_{i}(k-1)+\mu_{i}}\right)\right\| 
\\&\leqslant \sigma_{i}\left\|\tilde{\Phi}_{i}^{s}(k-1)\right\|\end{aligned}&&
\raisetag{2\baselineskip}
\end{flalign}

Given that both $\Phi_{i}(k)$ and then $\Phi_{i}^s(k)$ are bounded, it follows that $\Delta\Phi_{i}(k)$ and $\Delta\Phi_{i}^s(k)$ are also bounded, we obtain
\vspace{-5pt} \begin{equation}
\label{eq31}
    \left\|\Delta \Phi_{i}^{s}(k)\right\| \leqslant \bar{\theta}_{1}
\end{equation}
\begin{flalign}
\label{eq32}
\begin{gathered}
\Bigg\|\frac{\eta_{i1}\left(h_{i,r}(k)-1\right)\Phi_{i}^{s}(k-1) \Delta u_{i}(k-1) \Delta u_{i}^{\mathrm{T}}(k-1)}{\Delta u_{i}^{\mathrm{T}}(k-1)Q_i^{\hat{\Phi}}\Delta u_{i}(k-1)+\mu_{i}}\Bigg\| \leqslant \bar{\theta}_{2}
\end{gathered}
\end{flalign}
Using \eqref{eq31} and \eqref{eq32} yields
\begin{flalign}
\label{eq33}
\begin{aligned}
&\left\|\tilde{\Phi}_{i}^{s}(k)\right\|  \leqslant \sigma_{i}\left\|\tilde{\Phi}_{i}^{s}(k-1)\right\|+\left(\bar{\theta}_{1}+\bar{\theta}_{2}\right) \\&
\leqslant \sigma_{i}^{2}\left\|\tilde{\Phi}_{i}^{s}(k-2)\right\|+\sigma_{i}\left(\bar{\theta}_{1}+\bar{\theta}_{2}\right)+\left(\bar{\theta}_{1}+\bar{\theta}_{2}\right) \\&
 \leqslant \ldots 
 \leqslant \sigma_{i}^{k-1}\left\|\tilde{\Phi}_{i}^{s}(1)\right\|+\frac{1-\sigma_{i}^{k-1}}{1-\sigma_{i}}\left(\bar{\theta}_{1}+\bar{\theta}_{2}\right)
\end{aligned}\raisetag{3\baselineskip}
\end{flalign}

Hence, $\tilde{\Phi}_{i}^{s}(k)$ and then $\tilde{\Phi}_{i}(k)$ exhibit boundedness. Given that $\tilde{\Phi}_{i}(k)=\hat{\Phi}_{i}(k)-\Phi_{i}(k)$, the bound of $\hat{\Phi}_{i}(k)$ is guaranteed. That is $\left\|\hat{\Phi}_{i}(k)\right\| \leqslant \hat{\phi}_{i}$.

Define $\tilde{\xi}_i(k)=\xi_i(k)-\hat{\xi}_i(k)$ as the estimation error of the observer. According to \eqref{eq12} and \eqref{eq26}, we have
\begin{flalign}
\label{eq34}
\begin{aligned}
    &\tilde{\xi}_i(k+1) = \xi_i(k)+\Phi_i(k)\Delta u_i(k)
    \\&-\bigg(\hat{\xi}_{i}(k)+\hat{\Phi}_{i}(k) \Delta u_{i}(k)+K_{i}\left(\xi_{i}^c(k)-\hat{\xi}_{i}(k)\right)\bigg)
    \\&=\tilde{\xi}_i(k)-\tilde{\Phi}_i\Delta u_i(k)
    \\&-K_{i}\bigg(H_i(k)\Delta\xi_{i}(k)+\xi_{i}(k-1)-\hat{\xi}_{i}(k)\bigg)
    \\&=\tilde{\xi}_i(k)-\tilde{\Phi}_i\Delta u_i(k)
    -K_{i}\bigg(\left(H_i(k)-I\right)\Delta\xi_{i}(k)+\tilde{\xi}_{i}(k)\bigg)
    \\&=\left(I-K_i\right)\tilde{\xi}_i(k)-\tilde{\Phi}_i\Delta u_i(k)
    \\&+K_i\left(I-H_i(k)\right)
    \Phi_{i}(k-1)\Delta u_{i}(k-1)
\end{aligned}\raisetag{5\baselineskip}
\end{flalign}
Hence, given the established boundedness of $\tilde{\Phi}_{i}(k)$ and $\Delta u_{i}(k-1)$ in Remark \ref{rem:2}, there is a positive constant $\varsigma_i$ such that 
\vspace{-5pt} \begin{equation}
\label{eq35}
    \|\tilde{\Phi}_i\Delta u_i(k)\|+\|K_i\left(I-H_i(k)\right)\Phi_{i}(k-1)\Delta u_{i}(k-1)\|\leq \varsigma_i
\end{equation}

Taking the norm of \eqref{eq34} yields
\vspace{-5pt} \begin{equation}
\label{eq36}
\begin{aligned}
        &\left\|\tilde{\xi}(k+1)\right\|\leq \left\|I-K_i\right\|\left\|\tilde{\xi}_i(k)\right\|+\left\|\tilde{\Phi}_i\Delta u_i(k)\right\|
        \\&+\left\|K_i\left(I-H_i(k)\right)\Phi_{i}(k-1)\Delta u_{i}(k-1)\right\|
        \\&=\left\|I-K_i\right\|\left\|\tilde{\xi}_i(k)\right\|+\varsigma_i
        \\&\leq\left\|I-K_i\right\|^2\left\|\tilde{\xi}_i(k-1)\right\|+\left\|I-K_i\right\|\varsigma_i+\varsigma_i        \\&\leq\cdots\leq\left\|I-K_i\right\|^k\left\|\tilde{\xi}_i(1)\right\|+\frac{\varsigma_i\left(1-\left\|I-K_i\right\|^{k-1}\right)}{1-\left\|I-K_i\right\|}
    \end{aligned}
\end{equation}

Since $K_i=\mathrm{diag}(k_{i,r}), k_{i,r}\in(0,2),r=1,\cdots,p$, $\left\|I-K_i\right\|<1$, the boundedness of $\tilde{\xi}_i(k)$ is established. Using \eqref{eq27} to formulate \eqref{eq26} as
\begin{flalign}
\label{eq37}
\begin{aligned}
&\hat{\xi}_{i}(k+1)=\hat{\xi}_{i}(k)+K_{i}\left(\xi_{i}^c(k)-\hat{\xi}_{i}(k)\right)
\\&+\hat{\Phi}_{i}(k)\left(-\frac{\eta_{i2}\hat{\Phi}_{i}(k)^{\mathrm{T}}Q_i^u\bigg(\hat{\xi}_{i}(k)+K_{i}\left(\xi_{i}^c(k)-\hat{\xi}_{i}(k)\right)\bigg)}{\|\hat{\Phi}_{i}(k)^{\mathrm{T}}Q_i^u\hat{\Phi}_{i}(k)\|+ \|R_i^u\|}\right)
\\&=\Bigg(I-\frac{\eta_{i2}\hat{\Phi}_{i}(k)\hat{\Phi}_{i}(k)^{\mathrm{T}}Q_i^u}{\|\hat{\Phi}_{i}(k)^{\mathrm{T}}Q_i^u\hat{\Phi}_{i}(k)\|+ \|R_i^u\|}\Bigg)\times
\\&\bigg(\hat{\xi}_{i}(k)+K_{ i}\left(H_i(k)\Delta\xi_{i}(k)+\xi_{i}(k-1)-\hat{\xi}_{i}(k)\right)\bigg)
\\&=\left(I-\frac{\eta_{i2}\hat{\Phi}_{i}(k)\hat{\Phi}_{i}(k)^{\mathrm{T}}Q_i^u}{\|\hat{\Phi}_{i}(k)^{\mathrm{T}}Q_i^u\hat{\Phi}_{i}(k)\|+ \|R_i^u\|}\right)\times
\\&\left(\hat{\xi}_{i}(k)+K_{i}\left(\left(H_i(k)-I\right)\Phi_{i}(k-1)\Delta u(k-1)+\tilde{\xi}_{i}(k)\right)\right)
\end{aligned}\raisetag{7\baselineskip}
\end{flalign}

Let $\Upsilon_i(k) =I-\frac{\eta_{i2}\hat{\Phi}_{i}(k)\hat{\Phi}_{i}(k)^{\mathrm{T}}Q_i^u}{\|\hat{\Phi}_{i}(k)^{\mathrm{T}}Q_i^u\hat{\Phi}_{i}(k)\|+ \|R_i^u\|}$. According to Lemma \ref{lem:2} and triangle inequality, its Gershgorin disk follows
\begin{flalign}
\label{eq38}
\begin{aligned}
    D_{i,r}=&\left\{z_{i} \Bigg|\left| z_{i}\right|\leqslant\left|1-\frac{\eta_{i2} \sum_{s=1}^{q} \hat{\Phi}_{i, rs}^2(k) \varrho_{i,r}}{\|\hat{\Phi}_{i}(k)^{\mathrm{T}}Q_i^u\hat{\Phi}_{i}(k)\|+ \|R_i^u\|}\right|+\right.
    \\&\left.\left|\sum_{l=1, l \neq r}^{p} \frac{\eta_{i2} \sum_{s=1}^{q} \hat{\Phi}_{i, rs}(k) \hat{\Phi}_{i, ls}(k)\varrho_{i,r}}{\|\hat{\Phi}_{i}(k)^{\mathrm{T}}Q_i^u\hat{\Phi}_{i}(k)\|+ \|R_i^u\|}\right|\right\}
\end{aligned}\raisetag{2.5\baselineskip}
\end{flalign}

where $z_{i}$ is the characteristic root of $\Upsilon_i(k),r=1,\cdots,p,s=1,\cdots,q$. According to the proof of Theorem 2 in \cite{Hou2019a}, $z_{i}$ satisfies $\left|z_{i}\right|<1$ when $\left|\Phi_{i}(k)\right| \leqslant \phi_{i},\left|\hat{\Phi}_{i}(k)\right| \leqslant \hat{\phi}_{i}$. Hence, we have $\rho\left(\Upsilon_i(k)\right)<1$. Based on Lemma \ref{lem:3}, there is a positive constant $c_{i}$ such that
\vspace{-5pt} \begin{equation}
\label{eq39}
\begin{aligned}
\left\|\Upsilon_i(k)\right\|\leqslant\left\|\Upsilon_i(k)\right\|_{d} \leqslant 
\rho\left(\Upsilon_i(k)\right)+c_i
=\gamma_i<1
\end{aligned}
\end{equation}

Given the established boundedness of $\tilde{\xi}_{i}(k)$ and $\Delta u_{i}(k-1)$ in Remark \ref{rem:2}, there exists a positive constant $\varepsilon_i$ such that
\vspace{-5pt} \begin{equation}
\label{eq40}
\left\|K_{i}\left(\left(H_i(k)-I\right)\Phi_{i}(k-1)\Delta u(k-1)+\tilde{\xi}_{i}(k)\right)\right\| \leqslant \varepsilon_i
\end{equation}

Taking the norm of \eqref{eq37} and use \eqref{eq40}, we obtain
\begin{flalign}
\label{eq41}
\begin{aligned}
&\left\|\hat{\xi}_i(k+1)\right\|\leq\left\|\Upsilon_i(k)\right\|\left\|\hat{\xi_i}(k)\right\|+
\\&\left\|\Upsilon_i(k)\right\| \left\|K_{i}\left(\left(H_i(k)-I\right)\Phi_{i}(k-1)\Delta u(k-1)+\tilde{\xi}_{i}(k)\right)\right\|
\\&
\leqslant \gamma_i\|\hat{\xi}(k)\|+\gamma_i\varepsilon_i 
 \leqslant \gamma_i\|\hat{\xi}(k-1)\|+\gamma_i^2 \varepsilon_i+\gamma_i\varepsilon_i 
 \\&
 \leqslant \ldots 
 \leqslant \gamma_i^{k}\|\hat{\xi}(1)\|+\frac{\varepsilon_i\left(1-\gamma_i^{k-1}\right)}{1-\gamma_i}
\end{aligned}\raisetag{3\baselineskip}
\end{flalign}

That is, the boundedness of $\hat{\xi}_{i}(k)$ is guaranteed. Given the boundedness of $\hat{\xi}_{i}(k)$ and $\tilde{\xi}_{i}(k)$, the NABCE $\xi_{i}(k)$ is uniformly bounded. Then based on Lemma \ref{lem:4}, the local asymmetric bipartite consensus error $e_{y_i}$ is also bounded. Hence, the URABC problem in Definition \ref{prb:1} is solved.
\end{proof}
\section{Numerical Example}
\label{sec:simulation}
In this section, a numerical example is provided to validate the proposed results. Consider a communication digraph consisting of 6 agents shown in Fig.\ref{Fig.1}. Let $\mathscr{V}_{1}=\{1,2,4,6\}$ and $ \mathscr{V}_{2}=\{3,5\}$. The influence coefficients are $m=2$ and $n=4$. The dimension of input $u_i$ and output $y_i$ are $2$. The probability of the
successful DoS attacks are chosen as $\bar{h}_{i}=[0.2,0.3,0.24, 0.33,0.1,0.22],i=1,\cdots,6$. The models of the agents are \\$\begin{bmatrix}
    y_{i,1}(k+1)
    \\y_{i,2}(k+1)
\end{bmatrix}=
\begin{bmatrix}
    \frac{y_{i,1}(k) u_{i,1}(k)}{1+y_{i,1}(k)^{c_{i,1}}}+c_{i,2}u_{i,1}(k)
    \\\frac{y_{i,2}(k) u_{i,2}(k)}{1+y_{i,1}(k)^{c_{i,3}}+y_{i,2}(k)^{c_{i,3}}}+c_{i,4} u_{i,2}(k)
\end{bmatrix}$
where $k \in[1,2500]$, $c_{i,1}=[2,3,4,3,1,2]$, $c_{i,2}=[2,5,5,2,1,2]$, $ c_{i,3}=[1,0.9,0.6,1.1,1.3,1.5]$, $ c_{i,4}=[0.8,0.5,0.7,1.2,1.4,1.6],i=1,\cdots,6$. The desired reference signal is $y_d(k) = \begin{bmatrix}
        5
        &5
    \end{bmatrix}^T , k \in [0,1249]$ and $ y_d(k) =\begin{bmatrix}
        3
        &3
    \end{bmatrix}^T , k \in [1250,2500]$.
    
\begin{figure}[H]
\centering
\includegraphics[width=0.2\textwidth]{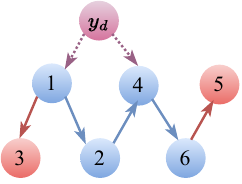}
\caption{The communication digraph of the networked MAS.} 
\label{Fig.1}
\end{figure}
The outputs of the agents using the proposed control protocols are profiled in Fig.\ref{Fig.2}. As seen, given $y_d=5$, $y_i$ for $i \in \mathscr{V}_{1}$ and $i \in \mathscr{V}_{2}$ converges to small boundaries around $10$ and $-20$, respectively. For $y_d=3$, the values are around $6$ and $-12$, respectively. This validates that by using the proposed DRMFAC algorithm, the URABC problem is solved for nonlinear MAS against DoS attacks. 
\begin{figure}[H]
\centering
\includegraphics[width=0.5 \textwidth]{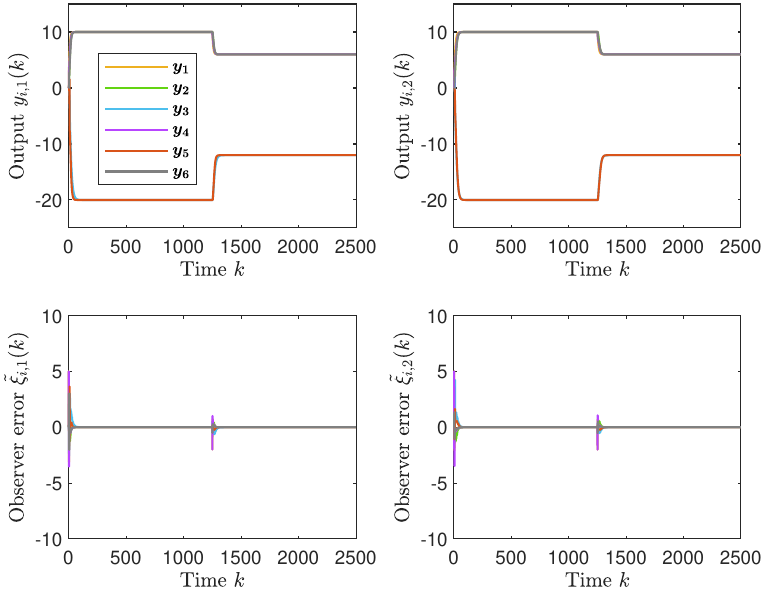}
\caption{Output trajectories $y_i$ and observer errors $\tilde{\xi}_i$ of the agents using the proposed control protocols.} 
\label{Fig.2}
\end{figure}

\section{CONCLUSION}

In this letter, we have addressed the URABC problem for nonlinear MAS under DoS attacks by our DRMFAC algorithm. We have proved that the URABC problem is solved by stabilizing the NABCE. Then, we have developed a DCFDL method to linearize the NABCE. Finally, we have used an EDSO to enhance the robustness against unknown dynamics and an attack compensation mechanism to eliminate the adverse effects of DoS attacks. 




\bibliographystyle{IEEEtran}
\bibliography{root}

\begin{thebibliography}{10}
\providecommand{\url}[1]{#1}
\csname url@samestyle\endcsname
\providecommand{\newblock}{\relax}
\providecommand{\bibinfo}[2]{#2}
\providecommand{\BIBentrySTDinterwordspacing}{\spaceskip=0pt\relax}
\providecommand{\BIBentryALTinterwordstretchfactor}{4}
\providecommand{\BIBentryALTinterwordspacing}{\spaceskip=\fontdimen2\font plus
\BIBentryALTinterwordstretchfactor\fontdimen3\font minus
  \fontdimen4\font\relax}
\providecommand{\BIBforeignlanguage}[2]{{%
\expandafter\ifx\csname l@#1\endcsname\relax
\typeout{** WARNING: IEEEtran.bst: No hyphenation pattern has been}%
\typeout{** loaded for the language `#1'. Using the pattern for}%
\typeout{** the default language instead.}%
\else
\language=\csname l@#1\endcsname
\fi
#2}}
\providecommand{\BIBdecl}{\relax}
\BIBdecl

\bibitem{olfati2004consensus}
R.~Olfati-Saber and R.~M. Murray, ``Consensus problems in networks of agents
  with switching topology and time-delays,'' \emph{IEEE Transactions on
  automatic control}, vol.~49, no.~9, pp. 1520--1533, 2004.

\bibitem{ren2007information}
W.~Ren, R.~W. Beard, and E.~M. Atkins, ``Information consensus in multivehicle
  cooperative control,'' \emph{IEEE Control systems magazine}, vol.~27, no.~2,
  pp. 71--82, 2007.

\bibitem{cao2012overview}
Y.~Cao, W.~Yu, W.~Ren, and G.~Chen, ``An overview of recent progress in the
  study of distributed multi-agent coordination,'' \emph{IEEE Transactions on
  Industrial informatics}, vol.~9, no.~1, pp. 427--438, 2012.

\bibitem{bu2017data}
X.~Bu, Z.~Hou, and H.~Zhang, ``Data-driven multiagent systems consensus
  tracking using model free adaptive control,'' \emph{IEEE transactions on
  neural networks and learning systems}, vol.~29, no.~5, pp. 1514--1524, 2017.

\bibitem{zhang2019networked}
X.-M. Zhang, Q.-L. Han, X.~Ge, D.~Ding, L.~Ding, D.~Yue, and C.~Peng,
  ``Networked control systems: A survey of trends and techniques,''
  \emph{IEEE/CAA Journal of Automatica Sinica}, vol.~7, no.~1, pp. 1--17, 2019.

\bibitem{hou2010novel}
Z.~Hou and S.~Jin, ``A novel data-driven control approach for a class of
  discrete-time nonlinear systems,'' \emph{IEEE Transactions on Control Systems
  Technology}, vol.~19, no.~6, pp. 1549--1558, 2010.

\bibitem{ren2020robust}
Y.~Ren and Z.~Hou, ``Robust model-free adaptive iterative learning formation
  for unknown heterogeneous non-linear multi-agent systems,'' \emph{IET Control
  Theory \& Applications}, vol.~14, no.~4, pp. 654--663, 2020.

\bibitem{wang2020consensus}
Y.~Wang, H.~Li, X.~Qiu, and X.~Xie, ``Consensus tracking for nonlinear
  multi-agent systems with unknown disturbance by using model free adaptive
  iterative learning control,'' \emph{Applied Mathematics and Computation},
  vol. 365, p. 124701, 2020.

\bibitem{guo2018asymmetric}
X.~Guo, J.~Liang, and J.~Lu, ``Asymmetric bipartite consensus over directed
  networks with antagonistic interactions,'' \emph{IET Control Theory \&
  Applications}, vol.~12, no.~17, pp. 2295--2301, 2018.

\bibitem{liang2021finite}
J.~Liang, X.~Bu, L.~Cui, and Z.~Hou, ``Finite time asymmetric bipartite
  consensus for multi-agent systems based on iterative learning control,''
  \emph{International Journal of Robust and Nonlinear Control}, vol.~31,
  no.~12, pp. 5708--5724, 2021.

\bibitem{altafini2012consensus}
C.~Altafini, ``Consensus problems on networks with antagonistic interactions,''
  \emph{IEEE transactions on automatic control}, vol.~58, no.~4, pp. 935--946,
  2012.

\bibitem{wolbrecht2008optimizing}
E.~T. Wolbrecht, V.~Chan, D.~J. Reinkensmeyer, and J.~E. Bobrow, ``Optimizing
  compliant, model-based robotic assistance to promote neurorehabilitation,''
  \emph{IEEE Transactions on Neural Systems and Rehabilitation Engineering},
  vol.~16, no.~3, pp. 286--297, 2008.

\bibitem{hussain2013adaptive}
S.~Hussain, S.~Q. Xie, and P.~K. Jamwal, ``Adaptive impedance control of a
  robotic orthosis for gait rehabilitation,'' \emph{IEEE transactions on
  cybernetics}, vol.~43, no.~3, pp. 1025--1034, 2013.

\bibitem{Xu2014}
D.~Xu, B.~Jiang, and P.~Shi, ``A novel model-free adaptive control design for
  multivariable industrial processes,'' \emph{IEEE Transactions on Industrial
  Electronics}, vol.~61, no.~11, pp. 6391--6398, 2014.

\bibitem{Hou2013}
Z.~Hou and S.~Jin, \emph{Model free adaptive control}.\hskip 1em plus 0.5em
  minus 0.4em\relax Boca Raton, FL, USA: CRC Press, 2013.

\bibitem{Hui2019}
Y.~Hui \emph{et~al.}, ``Extended state observer-based data-driven iterative
  learning control for permanent magnet linear motor with initial shifts and
  disturbances,'' \emph{IEEE Transactions on Systems, Man, and Cybernetics:
  Systems}, vol.~51, no.~3, pp. 1881--1891, 2019.

\bibitem{bell1965gershgorin}
H.~E. Bell, ``Gershgorin's theorem and the zeros of polynomials,'' \emph{The
  American Mathematical Monthly}, vol.~72, no.~3, pp. 292--295, 1965.

\bibitem{ortega2000iterative}
J.~M. Ortega and W.~C. Rheinboldt, \emph{Iterative solution of nonlinear
  equations in several variables}.\hskip 1em plus 0.5em minus 0.4em\relax SIAM,
  2000.

\bibitem{Hou2019a}
Z.~Hou and S.~Xiong, ``On model-free adaptive control and its stability
  analysis,'' \emph{IEEE Transactions on Automatic Control}, vol.~64, no.~11,
  pp. 4555--4569, 2019.

\end{thebibliography}

\end{document}